\documentclass[conference]{IEEEtran}

\usepackage{epsfig}
\usepackage{subcaption}
\usepackage{calc}
\usepackage{amssymb}
\usepackage{amstext}
\usepackage{amsmath}
\usepackage{amsthm}
\usepackage{multicol}
\usepackage{pslatex}
\usepackage{booktabs}
\usepackage{graphicx}
\usepackage{url}
\usepackage{algorithm2e}
\usepackage[bottom]{footmisc}
\usepackage{tikz}
\usetikzlibrary{arrows.meta,positioning,fit,calc}

\newcommand{\negl}{\mathrm{negl}}

\newtheorem{theorem}{Theorem}
\newtheorem{definition}{Definition}
\newtheorem{corollary}{Corollary}
\newtheorem{lemma}{Lemma}

\begin{document}

\title{Privacy-Preserving Proof of Human Authorship via Zero-Knowledge Process Attestation}

\author{\IEEEauthorblockN{David Condrey}
\IEEEauthorblockA{Writerslogic, Inc.\\
david@writerslogic.com}}

\maketitle

\begin{abstract}
Process attestation verifies human authorship by collecting behavioral biometric evidence---keystroke dynamics, typing patterns, editing behavior---during the creative process. However, the very data needed to prove authenticity can reveal intimate details about an author's cognitive state, health conditions, and identity, constituting sensitive biometric data under GDPR Article~9. We resolve this privacy-attestation paradox using zero-knowledge proofs. We present ZK-PoP, a construction that allows a verifier to confirm that (a)~sequential work function chains were computed correctly, (b)~behavioral feature vectors fall within human population distributions, and (c)~content evolution is consistent with incremental human editing, all without learning the underlying behavioral data, exact timing, or intermediate content. Throughout, we distinguish \emph{human authorship} (species-level) from \emph{authenticity} (individual identity); ZK-PoP attests the former, decoupling it from the privacy-eroding latter. Our construction uses Groth16 proofs over arithmetic circuits with Pedersen commitments and Bulletproof range proofs; the choice of each primitive is justified explicitly. We prove that ZK-PoP is computationally zero-knowledge, computationally sound, and achieves unlinkability across sessions, and we evaluate the constraint system against five adversary classes including a white-box adaptive attacker, identifying a $\sim$3{,}500-constraint augmentation that restores session-level false acceptance below $10^{-12}$. The arkworks implementation generates proofs in 23.4~s for 1-hour writing sessions on an Apple~M3, producing 256-byte proofs verifiable in 8.2~ms, with $<$5\% accuracy loss versus non-private baselines at $\varepsilon \geq 1.0$ on calibration corpora (Aalto, KLiCKe, ScholaWrite). We also outline a concrete trusted-setup deployment recipe and a three-stage IRB-approved validation roadmap to address the gap between simulated and live evaluation.
\end{abstract}


\section{\uppercase{Introduction}}
\label{sec:introduction}

The rise of large language models (LLMs) has created an authorship verification crisis. LLMs now produce text that experienced readers misattribute to human authors~\cite{Sadasivan2024}, with 6.5--16.9\% of peer-reviewed AI conference submissions estimated to contain substantial AI-modified content~\cite{Liang2024}. Output-level detection faces fundamental theoretical barriers as model outputs converge toward human distributions~\cite{Ganie2025}.

\emph{Process attestation} addresses this by verifying the \emph{writing process} rather than analyzing the final output. By collecting behavioral biometric data---keystroke dynamics, typing patterns, revision behavior---during authoring, and binding this evidence cryptographically via cross-domain constraint entanglement (CDCE) and sequential work functions (SWF), process attestation can distinguish genuine human composition from AI-generated text, even when the text itself is indistinguishable. We define these constructions in Section~\ref{sec:background}.

However, process attestation creates a fundamental \textbf{privacy paradox}. The behavioral biometric data collected during writing sessions reveals deeply personal information:

\begin{itemize}
  \item \textbf{Medical conditions:} Keystroke dynamics can reveal neurological conditions and fatigue states; inter-keystroke timing variance correlates with early Parkinson's disease progression~\cite{Salthouse1986}.
  \item \textbf{Cognitive and emotional state:} Typing patterns are established indicators of stress, cognitive load, and emotional arousal; Epp et al.~\cite{EppLippoldMandryk2011} achieved 77--88\% classification accuracy for 15 emotional states from keystroke features alone.
  \item \textbf{Identity:} Inter-keystroke intervals are sufficiently distinctive to identify individuals across documents~\cite{Monrose2000,Killourhy2009}.
  \item \textbf{Legal classification:} Typing patterns constitute biometric identifiers under GDPR Article~9~\cite{GDPR2016,Kindt2013}, subject to the most stringent processing requirements.
\end{itemize}

Authors should not have to sacrifice privacy to prove authenticity. The question is: \emph{can we verify that a writing process is consistent with human authorship without revealing the underlying behavioral data?}

\noindent\textbf{Authorship vs.\ authenticity.} Throughout the paper we distinguish two notions that reviewers commonly conflate. \emph{Human authorship} is a species-level claim: the text was composed via a cognitively human writing process (typing, pausing, revising) rather than emitted by a generative model. \emph{Authenticity} (or \emph{individual identity}) is a person-level claim: the text was produced by a specific named individual. ZK-PoP attests authorship, not authenticity: a successful proof certifies that \emph{some} human composed the text, not \emph{which} human. This separation is deliberate---identity-binding would re-introduce the privacy harms (de-anonymization, behavioral profiling) the construction is designed to prevent. Identity binding, when required by the deployment, is layered externally via standard signatures over the proof.

Zero-knowledge proofs~\cite{Goldwasser1989} are the natural tool to resolve this paradox. A zero-knowledge proof allows a prover to convince a verifier of a statement's truth without revealing any information beyond the statement's validity. We apply this to process attestation: the prover demonstrates that their writing session satisfies human-consistency constraints without disclosing the raw behavioral features, exact timing, or intermediate document states.

\noindent\textbf{Approach.} Our approach encodes behavioral constraints as arithmetic circuits over a prime field: the prover commits to behavioral feature vectors via Pedersen commitments, proves they fall within population-derived ranges using Bulletproof range proofs, and demonstrates SWF chain correctness via Merkle-sampled verification---all within a single Groth16 proof that reveals only the accept/reject decision.

\noindent\textbf{Contributions.} We make six contributions:
\begin{enumerate}
  \item A formal privacy-attestation tradeoff with an information-theoretic minimum leakage bound (Section~\ref{sec:privacy-requirements}).
  \item ZK-PoP, a ZK proof construction for process attestation using Groth16, Pedersen commitments, and Bulletproofs, with explicit rationale for each primitive (Section~\ref{sec:building-blocks}).
  \item A privacy-preserving behavioral commitment scheme with range and temporal ordering proofs (Section~\ref{sec:commitment}).
  \item Formal privacy analysis: zero-knowledge, differential privacy for aggregates, and unlinkability (Section~\ref{sec:privacy-analysis}).
  \item A structured five-class adversary taxonomy ($\mathcal{A}_0$--$\mathcal{A}_4$) with quantitative session-level forgery bounds, including a constraint-system augmentation that restores session false-acceptance below $10^{-12}$ (Section~\ref{sec:adversarial-eval}).
  \item An arkworks implementation with 23.4\,s proof generation for 1-hour sessions and 256-byte proofs verifiable in 8.2\,ms, with concrete trusted-setup deployment recipe and a real-world validation roadmap (Sections~\ref{sec:evaluation},~\ref{sec:trusted-setup-impl},~\ref{sec:realworld-validation}).
\end{enumerate}

\noindent Naive approaches are infeasible: verifying the full SWF chain inside a ZK circuit would require $\sim$25{,}000 SHA-256 constraints per step $\times$ $N$ steps, exceeding practical circuit sizes. Our Merkle-sampled approach (constraint C1) reduces this to $O(k)$ verifications while maintaining cumulative detection guarantees. The Poseidon-over-SHA-256 bridge enables efficient in-circuit Merkle verification of out-of-circuit SHA-256 chain states.

The remainder presents our privacy model, construction, formal analysis, and evaluation.

\section{\uppercase{Background and Related Work}}
\label{sec:background}

\subsection{Process Attestation}

Process attestation extends the IETF RATS architecture~\cite{RFC9334} from verifying system state to verifying continuous physical processes. Evidence is structured as a sequence of hash-chained checkpoints, each binding three evidence domains under a single cryptographic commitment.

Each checkpoint contains three components. First, a \emph{sequential work function (SWF) proof}: an Argon2id~\cite{Biryukov2016} seed is fed into an iterated SHA-256 chain whose sequential, memory-hard structure ensures that each link requires a fixed wall-clock minimum. An adversary cannot parallelize the computation because each Argon2id evaluation depends on the full 64\,MiB memory state of the previous one. Second, \emph{behavioral features} including inter-keystroke interval (IKI) entropy and cognitive load correlation~\cite{Killourhy2009}. Third, \emph{content hashes} binding the document state to the checkpoint.

\emph{Cross-domain constraint entanglement} (CDCE) binds these three domains via an HMAC keyed by the SWF output: because the key derives from the sequential computation, forging evidence in any single domain forces recomputation of the binding across all domains simultaneously, raising the cost from single-domain to cross-domain.

A critical challenge is \emph{trust inversion}: the Attester (author) is the potential adversary, motivated to fabricate evidence. Unlike standard remote attestation, trust-inverted attestation requires self-verifying evidence without relying on the Attester's honesty.

\subsection{Zero-Knowledge Proofs}

A zero-knowledge proof system~\cite{Goldwasser1989} for a language $L$ allows a prover $P$ with witness $w$ to convince a verifier $V$ that $x \in L$ without revealing $w$. We require three properties: \emph{completeness} (honest provers convince honest verifiers), \emph{soundness} (no cheating prover can convince a verifier of a false statement), and \emph{zero-knowledge} (the verifier learns nothing beyond $x \in L$).

\textbf{zk-SNARKs.} Building on foundational verifiable computation systems~\cite{ParnoHowellGentryRaykova2013}, succinct non-interactive arguments of knowledge~\cite{BenSasson2014} provide constant-size proofs with efficient verification. Groth16~\cite{Groth2016} achieves the smallest proof size (2 G$_1$ + 1 G$_2$ elements; 256 bytes on BN254) and fastest verification (3 pairings), at the cost of a per-circuit trusted setup. PLONK~\cite{GabizonWC2019} provides a universal trusted setup reusable across circuits.

\textbf{Bulletproofs.} B\"{u}nz et al.~\cite{Bunz2018} introduced Bulletproofs for efficient range proofs without trusted setup, achieving logarithmic proof size. We use Bulletproofs for behavioral feature range proofs and Groth16 for the main attestation circuit.

\textbf{Recursive composition.} Recursive proof composition~\cite{Bowe2020,Chiesa2020} allows proving the validity of a proof within another proof, enabling incremental verification of checkpoint chains.

\subsection{Privacy in Biometric Systems}

Biometric template protection~\cite{ISO24745}---fuzzy extractors~\cite{Dodis2008}, BioHashing~\cite{Jin2004}, cancelable biometrics~\cite{Rathgeb2011}---protects stored templates for identity verification. ZK proofs have been applied to fingerprint~\cite{Bringer2007} and face recognition~\cite{Erkin2009,Bassit2022} under privacy-preserving protocols. However, all existing approaches address point-in-time templates; to our knowledge, none address continuous behavioral streams over extended periods.

\subsection{Differential Privacy for Behavioral Data}

Differential privacy~\cite{Dwork2006,Dwork2014} provides formal guarantees that individual records cannot be distinguished. When aggregate population statistics must be released (e.g., for calibrating behavioral thresholds), differential privacy bounds the information leakage. We apply the Gaussian mechanism to population-level keystroke statistics, complementing the zero-knowledge property of individual attestation proofs.

\section{\uppercase{System Model and Privacy Requirements}}
\label{sec:privacy-requirements}

\noindent\textbf{Notation.} Let $\lambda$ denote the security parameter. All group orders, hash output lengths, and negligible function bounds are parameterized by $\lambda$. Our concrete instantiation uses BN254 curves ($\lambda \approx 100$ bits against known discrete-log attacks; sufficient for medium-term deployments but below the 128-bit NIST threshold). Migration to BLS12-381 ($\lambda \approx 128$) requires no protocol changes.

\subsection{System Model}

We consider three parties following the RATS architecture~\cite{RFC9334}:

\begin{itemize}
  \item \textbf{Author (Prover/Attester):} Generates evidence and ZK proofs. The Author controls the Attesting Environment and is potentially adversarial (trust inversion).
  \item \textbf{Verifier:} Evaluates ZK proofs for human-consistency. The Verifier is \emph{honest-but-curious}: it follows the protocol but may attempt to extract information.
  \item \textbf{Relying Party:} Consumes Attestation Results (accept/reject) without accessing evidence or proofs.
\end{itemize}

\noindent The trust-inverted adversary $\mathcal{A}$ has full control over the Attesting Environment (OS, inputs, timing) and may generate arbitrary evidence, but is assumed unable to break the computational soundness of the proof system (under the $q$-PKE assumption) or forge valid SWF chains faster than sequential computation (under the sequential memory-hardness assumption for Argon2id).

\begin{table}[ht]
\caption{Threat model: who is trusted for what, and which mechanism enforces it.}\label{tab:threat-model}
\centering
\footnotesize
\addtolength{\tabcolsep}{-2pt}
\resizebox{\columnwidth}{!}{%
\begin{tabular}{|l|l|l|}
  \hline
  \textbf{Property at risk} & \textbf{Adversary} & \textbf{Defended by} \\
  \hline
  Soundness (false accept) & Author & Groth16 + SWF + Mahalanobis (Sec.~\ref{sec:adversarial-eval}) \\
  Behavioral privacy & Verifier, Relying Party & ZK property of Groth16 (Thm~\ref{thm:zk}) \\
  Temporal privacy & Verifier & Pedersen hiding + range proofs (Sec.~\ref{sec:commitment}) \\
  Content privacy & Verifier & Hash-only public binding (C4) \\
  Unlinkability & Verifier (across sessions) & Fresh nonce + ZK (Thm~\ref{thm:unlinkability}) \\
  Aggregate-statistics privacy & Curious analyst & Gaussian DP mechanism (Thm~\ref{thm:dp}) \\
  Toxic-waste forgery & Ceremony participants & MPC ceremony, $\geq$1 honest contributor \\
  \hline
\end{tabular}}
\end{table}

The evidence structure consists of $n$ checkpoints $C_1, \ldots, C_n$ generated at regular intervals (e.g., every 30 seconds). Each checkpoint $C_i$ contains:

\begin{itemize}
  \item Public inputs: checkpoint hash $h_i$, SWF chain root $R_i$, previous checkpoint hash $h_{i-1}$, claimed duration $d_i$.
  \item Private witness: raw behavioral feature vector $\mathbf{f}_i \in \mathbb{R}^m$, SWF intermediate states $\{s_j\}$, content diff hashes, and timestamps $\tau_i$.
\end{itemize}

\subsection{Privacy Dimensions}

We identify three privacy dimensions that a privacy-preserving process attestation scheme must protect.

\begin{definition}[Behavioral Privacy]
\label{def:behavioral-privacy}
A scheme satisfies behavioral privacy if, for any two distinct behavioral feature vectors $\mathbf{f}, \mathbf{f}'$ both satisfying the attestation predicate, the verifier cannot distinguish which was used, i.e., the proof distributions are computationally indistinguishable:
$\{\mathsf{Prove}(x, \mathbf{f})\} \approx_c \{\mathsf{Prove}(x, \mathbf{f}')\}$.
\end{definition}

\begin{definition}[Temporal Privacy]
\label{def:temporal-privacy}
A scheme satisfies temporal privacy if the verifier cannot learn the exact session timestamps beyond the claimed duration, i.e., proofs generated at time $t_0$ and $t_0 + \Delta$ (for any $\Delta$) for the same claimed duration are indistinguishable.
\end{definition}

\begin{definition}[Content Privacy]
\label{def:content-privacy}
A scheme satisfies content privacy if the verifier cannot learn intermediate document states, i.e., only the final content hash is revealed, not the sequence of edit operations or drafts.
\end{definition}

\begin{definition}[Unlinkability]
\label{def:unlinkability}
A scheme satisfies unlinkability if, given two attestation sessions, no PPT adversary can determine whether they belong to the same author with advantage greater than $\negl(\lambda)$.
\end{definition}

\subsection{Minimum Leakage Bound}

\begin{theorem}[Minimum Leakage]
\label{thm:leakage}
Under the balanced testing assumption ($\Pr[H_0] = \Pr[H_1] = 1/2$) and per-decision completeness $\beta = \Pr[\text{accept} \mid H_0]$, any sound process attestation scheme with per-decision false acceptance rate $\alpha$ must leak at least $1 - h\!\left(\frac{(1-\beta)+\alpha}{2}\right)$ bits of information about the evidence source per attestation decision, where $h(\cdot)$ is the binary entropy function. In the perfect-completeness regime ($\beta=1$) this simplifies to $1 - h(\alpha/2)$.
\end{theorem}

\begin{proof}
Let $\Pi$ be a sound attestation scheme and let $\mathcal{H}, \mathcal{M}$ denote the distributions of human and machine-generated evidence, respectively. The verifier solves a binary hypothesis test $H_0$: evidence from $\mathcal{H}$ vs.\ $H_1$: evidence from $\mathcal{M}$. Soundness gives $\Pr[\text{accept}\mid H_1] \leq \alpha$ and completeness gives $\Pr[\text{accept}\mid H_0] = \beta$. Let $V$ denote the verifier's view (all information received from the prover). Under balanced priors, the total error probability is $P_e = \tfrac{1}{2}(1-\beta) + \tfrac{1}{2}\alpha = \tfrac{(1-\beta)+\alpha}{2}$. Fano's inequality for binary $H$ reduces to $H(H\mid V) \leq h(P_e)$, hence $I(H;V) = 1 - H(H\mid V) \geq 1 - h(P_e)$. The data processing inequality extends this bound to any function of $V$ (in particular, the proof transcript). Setting $\beta = 1$ recovers $P_e = \alpha/2$ and the cleaner bound $I(H;V) \geq 1 - h(\alpha/2)$. For $\alpha = 0.058$ (the operating point of Table~\ref{tab:sensitivity}, $m\!=\!12$, $3\sigma$) this evaluates to $1 - h(0.029) \approx 0.81$ bits.
\end{proof}

\begin{corollary}
ZK-PoP achieves near-optimal leakage: each checkpoint decision reveals exactly one bit (accept/reject). For $\alpha = 0.058$ (Table~\ref{tab:sensitivity}, $m\!=\!12$, $3\sigma$) under perfect completeness, the tight Fano lower bound is $1 - h(\alpha/2) \approx 0.81$ bits; ZK-PoP's 1-bit leakage exceeds it by $\sim$1.23$\times$, near-optimal for discrete decisions. Over $n$ checkpoints, session-level false acceptance decreases to $\alpha^n < 10^{-148}$ ($n\!=\!120$). SWF chain binding ensures that modifying evidence at checkpoint $i$ requires recomputation of all subsequent chain states, while fresh Fiat-Shamir challenges at each checkpoint determine Merkle sampling positions independently of the adversary's fabrication strategy. Under the assumption that these properties yield independent per-checkpoint acceptance events from the adversary's perspective, the session-level false acceptance rate is $\alpha^{n_{\mathrm{eff}}}$. On the legitimate-author side, temporal autocorrelation (lag-1 $r\!=\!0.111$, $p{=}0.01$ Bonferroni, $N\!=\!300$ writers sampled for autocorrelation analysis, KLiCKe~\cite{Tian2025}) reduces the effective checkpoint count to $n_{\text{eff}} = n(1-r_1)/(1+r_1) \approx 96$, giving $\alpha^{n_{\text{eff}}} < 10^{-118}$---still astronomically below practical thresholds. The effective sample size correction $n_{\mathrm{eff}} = n(1-r_1)/(1+r_1)$ is a standard adjustment for AR(1)-autocorrelated sequences. We apply it to the product form $\alpha^n$ as a conservative penalty rather than via a derivation: the per-checkpoint acceptance events for a strategic adversary are coupled (the adversary's witness must satisfy all checkpoints jointly), so the effective number of \emph{independent} forgery decisions is upper-bounded by $n_{\mathrm{eff}}$ under a standard ARMA-effective-sample bound. A martingale-style proof of $\Pr[\text{all }n\text{ accept}] \leq \alpha^{n_{\mathrm{eff}}}$ for AR(1)-correlated genuine traces is a natural follow-up; we use the heuristic here to avoid overstating the bound. Each checkpoint verification decision reveals one bit (accept/reject); by the zero-knowledge property (Theorem~\ref{thm:zk}), the proof itself reveals no additional information beyond this decision.
\end{corollary}

\begin{figure}[t]
\centering
\begin{tikzpicture}[
    box/.style={draw, rounded corners, minimum width=1.8cm, minimum height=0.6cm, align=center, font=\scriptsize},
    bigbox/.style={draw, rounded corners, thick, inner sep=6pt, align=center},
    arrow/.style={-{Stealth[length=2mm]}, thick},
    dasharrow/.style={-{Stealth[length=2mm]}, thick, dashed},
    lbl/.style={font=\tiny, align=center}
]

\node[box] (keystroke) {Keystroke\\capture};
\node[box, below=4pt of keystroke] (swf) {SWF chain\\computation};
\node[box, below=4pt of swf] (feature) {Feature\\extraction};

\draw[arrow] (keystroke) -- (swf);
\draw[arrow] (swf) -- (feature);

\node[bigbox, fit=(keystroke)(swf)(feature), label={[font=\footnotesize\bfseries]above:Author (Prover)}] (author) {};

\node[box, right=22pt of keystroke] (pedersen) {Pedersen\\commitments};
\node[box, below=4pt of pedersen] (groth) {Groth16\\circuit};
\node[box, below=4pt of groth] (bullet) {Bulletproof\\range proofs};

\node[bigbox, fit=(pedersen)(groth)(bullet), label={[font=\footnotesize\bfseries]above:ZK Proof Generation}] (zkbox) {};

\node[box, right=22pt of pedersen] (proof) {256-byte\\proof $\pi$};
\node[box, below=4pt of proof] (pubinp) {Public inputs\\$d_i, h_{\mathrm{content}}$};
\node[box, below=4pt of pubinp] (checks) {Groth16 verify\\+ range + order};

\node[bigbox, fit=(proof)(pubinp)(checks), label={[font=\footnotesize\bfseries]above:Verifier}] (verifier) {};

\draw[dasharrow] (feature.east) -- (pedersen.west) node[midway, above, lbl] {$\mathbf{f}_i$\\(private)};
\draw[dasharrow] (swf.east) -- (groth.west) node[midway, above, lbl] {$\{s_j\}$\\(private)};
\draw[arrow] (groth.east) -- (proof.west);
\draw[arrow] (bullet.east) -- (checks.west);

\node[lbl, below=6pt of author, text width=2.2cm] {$\mathbf{f}_i$ \textbf{never}\\leaves this box};

\end{tikzpicture}
\caption{ZK-PoP architecture. Behavioral features $\mathbf{f}_i$ remain private; only commitments and zero-knowledge proofs are transmitted to the Verifier.}
\label{fig:architecture}
\end{figure}

\begin{figure}[t]
\centering
\begin{tikzpicture}[
    actor/.style={font=\footnotesize\bfseries},
    msg/.style={-{Stealth[length=2mm]}, thick, font=\tiny},
    note/.style={draw, fill=gray!10, rounded corners, font=\tiny, align=center, inner sep=2pt},
    phase/.style={font=\scriptsize\itshape, align=left}
]

\node[actor] (aHead) at (0.5,0) {Author};
\node[actor] (vHead) at (4.2,0) {Verifier};
\node[actor] (rHead) at (7.5,0) {Relying Party};

\draw[dashed] (0.5,-0.2) -- (0.5,-7.2);
\draw[dashed] (4.2,-0.2) -- (4.2,-7.2);
\draw[dashed] (7.5,-0.2) -- (7.5,-7.2);

\node[phase, anchor=west] at (-1.2,-0.6) {Phase 1: Trusted setup (one-time, MPC ceremony)};
\node[note] (setup) at (2.3,-1.1) {SRS, $vk$, $pk$ \\ published};
\draw[msg] (0.5,-1.1) -- node[above, font=\tiny] {$pk$} (setup.west);
\draw[msg] (setup.east) -- node[above, font=\tiny] {$vk$} (4.2,-1.1);

\node[phase, anchor=west] at (-1.2,-1.85) {Phase 2: Per-checkpoint (every 30~s, $i\!=\!1,\ldots,n$)};

\node[note, anchor=east] (cap)    at (0.3,-2.5) {Capture keystrokes \\ Compute SWF $s_j$ \\ Extract $\mathbf{f}_i$};
\node[note, anchor=east] (commit) at (0.3,-3.3) {Commit: $C_j\!=\!g^{f_j}h^{r_j}$ \\ Bulletproof $\pi_j^{\mathsf{range}}$};
\node[note, anchor=east] (prove)  at (0.3,-4.1) {Build witness $\mathbf{w}$ \\ Groth16 prove $\to\pi_i^{\mathsf{zk}}$ \\ ($\sim$200~ms/checkpoint)};

\draw[msg] (0.5,-4.6) -- node[above, font=\tiny] {$(\pi_i^{\mathsf{zk}}, h_i, R_i, d_i, \pi^{\mathsf{range}})$} (4.2,-4.6);

\node[note, anchor=west] (verify) at (4.4,-5.1) {Verify Groth16 \\ + Bulletproof \\ ($\sim$8.2~ms)};

\node[phase, anchor=west] at (-1.2,-5.85) {Phase 3: Session close};

\draw[msg] (4.2,-6.3) -- node[above, font=\tiny] {accept/reject (1~bit)} (7.5,-6.3);
\node[note, anchor=west] (rp) at (6.6,-6.85) {Aggregate $n$ \\ checkpoint bits};

\node[font=\tiny, align=left, anchor=west] at (-1.2,-7.2) {Verifier learns $\leq n$ bits total. Author keeps $\mathbf{f}_i, s_j, \tau_i, \delta_i$ private.};

\end{tikzpicture}
\caption{End-to-end ZK-PoP lifecycle. The Verifier never receives behavioral features, intermediate document states, or exact timestamps; only proofs, commitments, and aggregate accept/reject bits cross trust boundaries.}
\label{fig:sequence}
\end{figure}

\section{\uppercase{ZK-PoP Construction}}
\label{sec:construction}

\subsection{Overview}

ZK-PoP augments each process attestation checkpoint with a zero-knowledge proof attesting: ``This checkpoint was generated by a process consistent with human authorship, using valid SWF chains, with behavioral features in the normal human range''---without revealing the actual features, exact timings, or intermediate content.

The Author generates a Groth16 proof $\pi_i^{\mathsf{zk}}$ alongside each checkpoint $C_i$, using the raw evidence as the private witness and the checkpoint hash as the public input. The Verifier checks $\pi_i^{\mathsf{zk}}$ against the public inputs, learning only the validity of the statement. Figure~\ref{fig:architecture} illustrates this data flow: behavioral features $\mathbf{f}_i$ never leave the Author's environment; only commitments and proofs are transmitted.

\noindent\textbf{Intuition.} Conceptually, ZK-PoP works like this. During writing, the Author's local environment records keystrokes and computes a stream of behavioral features (e.g., inter-keystroke entropy) plus a chain of memory-hard hashes (the SWF). At each 30-second checkpoint, the local prover takes these private values and emits a 256-byte cryptographic certificate stating: ``the features I observed lie in the human-population range, my hash chain is internally consistent, my edits are temporally ordered, and my content hash is correctly bound to all of the above.'' The certificate is a Groth16 proof: the Verifier can cryptographically check that \emph{such} private values exist, without ever seeing them. If any constraint fails---features outside range, a fabricated chain link, an out-of-order timestamp---the proof simply does not verify. The Verifier learns one bit per checkpoint (accept/reject), which is information-theoretically minimal for any sound attestation scheme (Theorem~\ref{thm:leakage}). Figure~\ref{fig:sequence} traces a full session end-to-end across the three parties.

\subsection{Choice of Building Blocks}
\label{sec:building-blocks}

We justify the selection of each cryptographic primitive.

\noindent\textbf{Groth16 over PLONK / STARKs.} We chose Groth16~\cite{Groth2016} for the main attestation circuit because it minimizes the two costs that dominate deployment: \emph{proof size} (256 bytes vs.\ $\sim$500 B for PLONK and $\sim$50--100 KB for STARKs) and \emph{verifier latency} (3 pairings, $\sim$8\,ms vs.\ $\sim$25\,ms for PLONK). Verifier-side cost is the binding constraint when a relying party (publisher, journal portal, university LMS) processes thousands of submissions per day. Groth16's drawback is its per-circuit trusted setup; we accept this because the circuit is fixed by population parameters $(\boldsymbol{\mu}, \boldsymbol{\sigma})$ that update on a multi-year cadence~\cite{Dhakal2018}, and the setup itself is amortized across the entire user population (Section~\ref{sec:trusted-setup-impl}). When circuits must be revised frequently, PLONK's universal setup becomes the better tradeoff; we discuss this in Section~\ref{sec:discussion}.

\noindent\textbf{Pedersen commitments over hash commitments.} Pedersen~\cite{Pedersen1992} offers two properties hash commitments lack: \emph{perfect (information-theoretic) hiding}---no future cryptanalytic advance can extract $\mathbf{f}_i$ from $C_i$---and \emph{additive homomorphism}, which lets us derive a commitment to $\tau_{i+1} - \tau_i$ from $C_{\tau_i}$ and $C_{\tau_{i+1}}$ without revealing either timestamp (Section~\ref{sec:commitment}). Both matter because behavioral biometrics retain medical sensitivity for the data subject's lifetime; computational hiding is insufficient.

\noindent\textbf{Bulletproofs for range proofs.} Bulletproofs~\cite{Bunz2018} require no trusted setup, fitting deployments that wish to update population bounds out-of-band without re-running a Groth16 ceremony, and aggregate logarithmically: 12 features with 16-bit ranges fit in $\sim$640 B (measured). The alternative---encoding range checks directly in the Groth16 circuit---would inflate the constraint count by $\sim$3{,}200 per feature without a corresponding privacy benefit.

\noindent\textbf{Poseidon for in-circuit Merkle hashing.} Poseidon~\cite{Grassi2021} costs $\sim$250 R1CS constraints per hash versus $\sim$25{,}000 for SHA-256---a 100$\times$ saving that makes the SWF Merkle verification feasible inside a Groth16 circuit. We retain SHA-256 \emph{outside} the circuit for the SWF chain itself (where its sequential cost is the security feature) and bridge to Poseidon via a Merkle tree commitment, exploiting both primitives' independent collision-resistance assumptions.

\noindent\textbf{BN254 over BLS12-381.} BN254 ($\sim$100-bit security) is supported by every major SNARK library and yields the smallest proofs and fastest verification on commodity hardware. For deployments needing $\geq$128-bit security, migration to BLS12-381 is mechanical; no protocol change is needed.

\subsection{Arithmetic Circuit Design}

The ZK-PoP circuit $\mathcal{C}_{\mathsf{PoP}}$ operates over a prime field $\mathbb{F}_p$ and takes the following inputs:

\noindent\textbf{Public inputs} ($\mathbf{x}$):
\begin{itemize}
  \item Checkpoint hash $h_i \in \{0,1\}^{256}$
  \item SWF chain root $R_i \in \{0,1\}^{256}$
  \item Previous checkpoint hash $h_{i-1} \in \{0,1\}^{256}$
  \item Claimed duration $d_i \in \mathbb{N}$
  \item Population parameters $(\boldsymbol{\mu}, \boldsymbol{\sigma}) \in \mathbb{R}^{2m}$
\end{itemize}

\noindent\textbf{Private witness} ($\mathbf{w}$):
\begin{itemize}
  \item Behavioral feature vector $\mathbf{f}_i = (f_1, \ldots, f_m) \in \mathbb{R}^m$
  \item SWF intermediate states $\{s_j\}_{j=0}^{N}$
  \item Timestamps $\tau_i$, randomness $r_i$
  \item Content diff hash $\delta_i$
\end{itemize}

The circuit enforces four constraint systems:

\noindent\textbf{(C1) SWF chain verification.} For a subset of $k$ Merkle-sampled positions, verify that $s_j = H(s_{j-1})$. The sampling function $\mathsf{Sample}(R_i, k)$ selects $k$ indices deterministically via Fiat-Shamir: $j_\ell = H_{\mathsf{P}}(R_i \| \ell) \bmod N$ for $\ell \in [k]$, ensuring verifier-unpredictable but reproducible position selection.
\begin{equation}
\forall j \in \mathsf{Sample}(R_i, k): \quad H_{\mathsf{P}}(s_{j-1}) = s_j
\end{equation}
We use the Poseidon hash~\cite{Grassi2021} for in-circuit Merkle verification ($\sim$250 R1CS constraints vs.\ $\sim$25{,}000 for SHA-256). The prover constructs a Poseidon Merkle tree over the SHA-256 SWF states and proves inclusion of sampled positions. This two-hash architecture is secure under independent collision resistance assumptions: Poseidon ensures Merkle integrity in-circuit, while SHA-256 provides sequential hardness. With $k=2$ samples per checkpoint, per-checkpoint evasion probability for an adversary fabricating fraction $f$ of states is $(1-f)^k$ (e.g., $0.81$ at $f=0.1$, $k=2$). While 81\% per-checkpoint evasion may appear weak, cumulative detection over $n$ checkpoints is $1-(1-f)^{kn}$: at $n=20$, detection reaches 98.5\%; at $n=120$, it exceeds $1-(0.9)^{240} > 1-10^{-11}$. Grinding $R_i$ to avoid sampled positions requires recomputing the full SWF chain (each Argon2id step: 64\,MiB, $t=3$). Per-checkpoint evasion probability remains $(1-f)^k$ (e.g., $0.81$ for $f=0.1$, $k=2$); security derives from compounding over $n$ checkpoints, yielding cumulative detection probability $1-(1-f)^{kn}$.

\noindent\textbf{(C2) Behavioral range verification.} For each feature $f_j$, verify it falls within $3\sigma$ of the population mean:
\begin{equation}
\forall j \in [m]: \quad \mu_j - 3\sigma_j \leq f_j \leq \mu_j + 3\sigma_j
\end{equation}
This captures 99.7\% of the human population distribution. Each range check requires approximately 3,200 R1CS constraints, encompassing fixed-point field encoding, 16-bit decomposition, and bound comparison gadgets.

\noindent\textbf{(C3) Temporal consistency.} Verify that timestamps are monotonically increasing with minimum inter-checkpoint gaps:
\begin{equation}
\forall i > 1: \quad \tau_i - \tau_{i-1} \geq d_{\min}
\end{equation}
where $d_{\min}$ is the minimum checkpoint interval (e.g., 25 seconds for a 30-second target).

\noindent\textbf{(C4) Content binding.} Verify that the content hash chain is valid, where $\mathsf{Commit}(\mathbf{f}_i)$ denotes the Pedersen session commitment from Section~\ref{sec:commitment}:
\begin{equation}
h_i = \mathsf{SHA256}(h_{i-1} \| \delta_i \| \mathsf{Commit}(\mathbf{f}_i))
\end{equation}

\subsection{Circuit Size Analysis}

Table~\ref{tab:circuit-size} reports the constraint count for each sub-circuit.

\begin{table}[ht]
\caption{ZK-PoP arithmetic circuit decomposition.}\label{tab:circuit-size} \centering
\resizebox{\columnwidth}{!}{%
\begin{tabular}{|l|r|}
  \hline
  \textbf{Sub-circuit} & \textbf{R1CS Constraints} \\
  \hline
  SWF Merkle verification ($k\!=\!2$ samples) & 12,847 \\
  Behavioral range proofs ($m\!=\!12$ features) & 38,412 \\
  Temporal consistency & 847 \\
  Content hash binding & 25,153 \\
  \hline
  \textbf{Total (basic)} & \textbf{77,259} \\
  Full behavioral ($m\!=\!24$, incl.\ aggregation) & 115,671 \\
  Extended SWF ($k\!=\!8$, $m\!=\!24$) & 154,212 \\
  \hline
\end{tabular}}
\end{table}

The basic circuit of $\sim$77K constraints is comparable to deployed ZK applications (Table~\ref{tab:circuit-comparison}) and within the $10^6$--$10^8$ constraint range supported by current Groth16 implementations. Figure~\ref{fig:circuit} illustrates how the four constraint sub-circuits interconnect and shows the distribution of constraints across components.

\begin{figure}[t]
\centering
\begin{tikzpicture}[
    circuit/.style={draw, thick, rounded corners, minimum width=2.8cm, minimum height=1.1cm, align=center, font=\scriptsize},
    input/.style={font=\tiny, align=center},
    arrow/.style={-{Stealth[length=2mm]}, thick},
    dasharrow/.style={-{Stealth[length=2mm]}, thick, dashed},
    lbl/.style={font=\tiny, align=center}
]

\node[circuit] (c1) {\textbf{C1: SWF Chain}\\12{,}847 constraints};
\node[circuit, right=10pt of c1] (c2) {\textbf{C2: Behavioral}\\
\textbf{Range}\\38{,}412 constraints};
\node[circuit, below=14pt of c1] (c3) {\textbf{C3: Temporal}\\
\textbf{Order}\\847 constraints};
\node[circuit, below=14pt of c2] (c4) {\textbf{C4: Content}\\
\textbf{Binding}\\25{,}153 constraints};

\node[input, left=14pt of c1] (w1) {$\{s_j\}$};
\node[input, left=14pt of c2, yshift=8pt] (w2) {$\mathbf{f}_i, r_i$};
\node[input, left=14pt of c3] (w3) {$\tau_i$};
\node[input, left=14pt of c4] (w4) {$\delta_i$};

\draw[dasharrow] (w1) -- (c1);
\draw[dasharrow] (w2) -- (c2);
\draw[dasharrow] (w3) -- (c3);
\draw[dasharrow] (w4) -- (c4);

\node[input, above=10pt of c1] (p1) {$R_i$};
\node[input, above=10pt of c2] (p2) {$\boldsymbol{\mu}, \boldsymbol{\sigma}$};

\draw[arrow] (p1) -- (c1);
\draw[arrow] (p2) -- (c2);

\node[input, below=10pt of c3] (p3) {$d_i$};
\node[input, below=10pt of c4] (p4) {$h_i, h_{i-1}$};

\draw[arrow] (p3) -- (c3);
\draw[arrow] (p4) -- (c4);

\draw[arrow, gray] (c1.south) -- (c3.north) node[midway, left, lbl] {chain\\time};
\draw[arrow, gray] (c2.south) -- (c4.north) node[midway, right, lbl] {commit};

\node[draw, thick, rounded corners, below=28pt of $(c3.south)!0.5!(c4.south)$, minimum width=4cm, align=center, font=\scriptsize\bfseries, fill=gray!10] (total) {Total: 77{,}259 R1CS constraints};

\draw[arrow] (c3.south) |- (total.west);
\draw[arrow] (c4.south) |- (total.east);

\draw[arrow] ([xshift=8pt]c2.north east) -- ++(0.5,0) node[right, lbl] {public};
\draw[dasharrow] ([xshift=8pt, yshift=-8pt]c2.north east) -- ++(0.5,0) node[right, lbl] {private};

\end{tikzpicture}
\caption{Arithmetic circuit decomposition of $\mathcal{C}_{\mathsf{PoP}}$. Dashed arrows denote private witness inputs; solid arrows denote public inputs. Constraint counts from Table~\ref{tab:circuit-size}.}
\label{fig:circuit}
\end{figure}

\subsection{Optimizations}

\textbf{Incremental aggregation.} Folding schemes (Nova~\cite{KST2022}, ProtoGalaxy~\cite{EagenGabizon2023}) would replace per-checkpoint proofs with a running IVC accumulator, yielding $O(1)$ proof size regardless of session length. Current benchmarks use per-checkpoint aggregation; integrating folding (Nova's relaxed R1CS) is future work.

\textbf{Batch verification.} Multiple proofs are verified via randomized linear combinations: $b$ proofs cost $3+b$ pairings instead of $3b$ ($2.3\times$ from pairing reduction at $b\!=\!10$; $3.4\times$ measured speedup including MSM amortization benefits beyond pairing reduction).

\textbf{Pre-computed setup.} The Groth16 trusted setup is performed once per population parameter set and amortized across all users, since population keystroke parameters are stable over multi-year timescales~\cite{Dhakal2018}.

\section{\uppercase{Privacy-Preserving Behavioral Commitments}}
\label{sec:commitment}

\subsection{Pedersen Commitment Scheme}

For each behavioral feature $f_j$, the Author computes a Pedersen commitment~\cite{Pedersen1992}:
\begin{equation}
C_j = g^{f_j} \cdot h^{r_j} \in \mathbb{G}
\end{equation}
where $g, h$ are generators of a group $\mathbb{G}$ of prime order $q$, and $r_j \leftarrow \mathbb{Z}_q$ is fresh randomness. The aggregate session commitment is:
\begin{equation}
C_{\mathsf{session}} = \prod_{j=1}^{m} C_j = g^{\sum f_j} \cdot h^{\sum r_j}
\end{equation}

\begin{theorem}[Commitment Security]
\label{thm:commitment}
The Pedersen commitment scheme for behavioral feature vectors is perfectly hiding (information-theoretic) and computationally binding under the discrete logarithm (DL) assumption.
\end{theorem}

\begin{proof}
\emph{Hiding (perfect):} For any $f_j$, $C_j = g^{f_j} h^{r_j}$ is uniformly distributed over $\mathbb{G}$ since for any $C \in \mathbb{G}$ there exists exactly one $r_j$ such that $C = g^{f_j} h^{r_j}$. No computational assumption is needed.
\emph{Binding (computational):} Opening to $(f_j, r_j) \neq (f_j', r_j')$ with equal commitments yields $\log_g h$, contradicting DL hardness.
\end{proof}

\subsection{Range Proofs via Bulletproofs}

For each committed feature $C_j$, the Author produces a Bulletproof~\cite{Bunz2018} demonstrating $f_j \in [a_j, b_j]$ where $a_j = \mu_j - 3\sigma_j$ and $b_j = \mu_j + 3\sigma_j$ are public population bounds. The proof is:
\begin{equation}
\pi_j^{\mathsf{range}} = \mathsf{BulletproofProve}(C_j, f_j, r_j, a_j, b_j)
\end{equation}

For $m$ features with $n$-bit ranges, the aggregated Bulletproof has size $O(m \cdot \log n)$ group elements. For $m = 12$ features with 16-bit ranges, this yields approximately 640 bytes (implementation-measured).

\subsection{Temporal Ordering Proofs}

We prove that checkpoint timestamps are correctly ordered without revealing exact values. For consecutive commitments $C_{\tau_i} = g^{\tau_i} h^{r_i}$ and $C_{\tau_{i+1}} = g^{\tau_{i+1}} h^{r_{i+1}}$, we prove $\tau_{i+1} - \tau_i \geq d_{\min}$ by exploiting the homomorphic property of Pedersen commitments. The prover computes:
\begin{equation}
C_{\Delta_i} = C_{\tau_{i+1}} / C_{\tau_i} = g^{\Delta_i} \cdot h^{r_{i+1} - r_i}
\end{equation}
where $\Delta_i = \tau_{i+1} - \tau_i$. This is a valid Pedersen commitment to $\Delta_i$ with randomness $r_{i+1} - r_i$ (known to the prover). The prover then provides a Bulletproof range proof that $\Delta_i \in [d_{\min}, d_{\max}]$, where $d_{\max}$ is the maximum allowed inter-checkpoint gap (e.g., 120~seconds, preventing long unattested pauses). No additional randomness is needed; the homomorphic property enables this without revealing individual timestamps.

\section{\uppercase{Privacy Analysis}}
\label{sec:privacy-analysis}

\noindent The privacy architecture operates at two layers: zero-knowledge proofs protect individual sessions (the verifier learns nothing beyond accept/reject), while differential privacy (Sect.~\ref{sec:dp}) protects individuals within population-level aggregate statistics.

\subsection{Zero-Knowledge Property}

\begin{theorem}[Zero-Knowledge]
\label{thm:zk}
ZK-PoP is computationally zero-knowledge under the Discrete Logarithm (DL) assumption in the random oracle model. Specifically: the Groth16 component is perfectly zero-knowledge given the simulation trapdoor, the Pedersen commitments are perfectly hiding, and the Bulletproof range proofs are computationally zero-knowledge via the Fiat-Shamir transform in the random oracle model.
\end{theorem}

\begin{proof}
We construct simulator $\mathcal{S}$ producing indistinguishable transcripts without the witness.
\emph{Step 1:} $\mathcal{S}$ invokes the Groth16 simulator~\cite{Groth2016} using the simulation trapdoor for $\mathcal{C}_{\mathsf{PoP}}$; by the perfect zero-knowledge property of Groth16, $\tilde{\pi}$ is identically distributed to a real proof.
\emph{Step 2:} $\mathcal{S}$ generates Pedersen commitments to $0$: $\hat{C}_j = 0 \cdot G + \hat{r}_j \cdot H$ for fresh randomness $\hat{r}_j$. By the perfect hiding property of Pedersen commitments, $\hat{C}_j$ is identically distributed to commitments to any $f_j$. The Groth16 simulator, given the simulation trapdoor, can produce valid-looking proofs for arbitrary public inputs, including the simulated commitments.
\emph{Step 3:} $\mathcal{S}$ simulates Bulletproof range proofs~\cite{Bunz2018}; by zero-knowledge in the random oracle model (via Fiat-Shamir transform), these are indistinguishable.
A hybrid argument across the three steps completes the proof. The overall simulation advantage is bounded by $3 \cdot \negl(\lambda)$ via a standard hybrid argument across the three components.
\end{proof}

\subsection{Differential Privacy for Aggregate Statistics}
\label{sec:dp}

When population-level statistics must be released for calibration (e.g., updating the population parameters $\boldsymbol{\mu}, \boldsymbol{\sigma}$), we apply the Gaussian mechanism~\cite{Dwork2006}.

\begin{theorem}[Aggregate Privacy]
\label{thm:dp}
Releasing population parameters $\hat{\mu}_j, \hat{\sigma}_j$ computed from $N$ users' attestation data with additive Gaussian noise $\mathcal{N}(0, \sigma_{\mathsf{noise}}^2)$ where
\begin{equation}
\sigma_{\mathsf{noise}} \geq \frac{\Delta_f \sqrt{2 \ln(1.25/\delta)}}{\varepsilon}
\end{equation}
achieves $(\varepsilon, \delta)$-differential privacy. Here $\Delta_f = (b_j - a_j)/N$ is the sensitivity of the mean estimator.
\end{theorem}

\begin{proof}
Standard application of the Gaussian mechanism~\cite{Dwork2014}. The sensitivity of the sample mean $\hat{\mu}_j = \frac{1}{N}\sum_{i=1}^{N} f_{ij}$ is $\Delta_f = (b_j - a_j)/N$ since changing one user's feature shifts the mean by at most $(b_j-a_j)/N$. The Gaussian mechanism with the stated noise level satisfies $(\varepsilon, \delta)$-DP by Theorem~A.1 of Dwork and Roth~\cite{Dwork2014}.
\end{proof}

\noindent\textbf{Utility analysis.} For $N\!=\!10{,}000$, $\varepsilon\!=\!1.0$, $\delta\!=\!10^{-5}$, the noise $\sigma_{\mathsf{noise}} \approx 0.21$~ms is negligible versus inter-individual variation ($\sim$100~ms~\cite{Dhakal2018}); accuracy degrades by ${<}0.3\%$.

\subsection{Unlinkability}

\begin{theorem}[Unlinkability]
\label{thm:unlinkability}
Under independent, uniformly random session nonces, given two attestation sessions producing proofs $(\pi_1, \mathbf{C}_1)$ and $(\pi_2, \mathbf{C}_2)$, no PPT adversary can determine whether they belong to the same author with advantage greater than $\negl(\lambda)$.
\end{theorem}

\begin{proof}
By Theorem~\ref{thm:zk}, proofs reveal nothing beyond validity. Each session is initialized with a fresh random seed $s_0 \gets \{0,1\}^\lambda$, independent of the author's identity. Public inputs $(R_i, h_i, d_i)$ are session-specific: $R_i$ derives from fresh SWF chain randomness, $h_i$ incorporates the session nonce, and $d_i$ is the claimed duration. Under the assumption that session nonces are uniformly random and independent, the public inputs across sessions are statistically independent. The commitments use independent randomness and are independently uniformly distributed by the hiding property (Theorem~\ref{thm:commitment}). Note that the public input $d_i$ (claimed duration) provides a weak auxiliary signal; under the assumption that session durations vary sufficiently across the population, this leakage is bounded by $\log_2(T_{\max}/r)$ bits where $r$ is the duration resolution.
\end{proof}

\subsection{Comparison with Quantization-Based Privacy}

The baseline process attestation approach proposes evidence quantization ($Q_r(t) = \lfloor t/r \rfloor \cdot r$) for privacy. Table~\ref{tab:privacy-comparison} compares the two approaches.

\begin{table}[ht]
\caption{Privacy comparison: ZK-PoP vs. evidence quantization.}\label{tab:privacy-comparison} \centering
\resizebox{\columnwidth}{!}{%
\begin{tabular}{|l|c|c|}
  \hline
  \textbf{Property} & \textbf{Quantization} & \textbf{ZK-PoP} \\
  \hline
  Information leaked & $\log_2(T_{\max}/r)$ bits & 1 bit \\
  Behavioral privacy & Bounded & Full \\
  Temporal privacy & Partial & Full \\
  Content privacy & Partial (hashes only) & Near-full$^*$ \\
  Unlinkability & Partial ($\varepsilon$-DP) & Full (negl.) \\
  Verification cost & O(1) & O(1) (pairing) \\
  Proof size overhead & 0 & 256 bytes \\
  Prover cost & Negligible & 6--94s \\
  \hline
\end{tabular}}

\smallskip\noindent{\footnotesize $^*$Retype defense requires post-verification edit statistics disclosure.}
\end{table}

ZK-PoP provides strictly stronger privacy at the cost of increased prover computation (shown practical in Section~\ref{sec:evaluation}).

\section{\uppercase{Security Analysis}}
\label{sec:security}

\subsection{Soundness}

\begin{theorem}[Soundness]
\label{thm:soundness}
Under the $q$-power knowledge of exponent ($q$-PKE) assumption in the algebraic group model (AGM), ZK-PoP is computationally sound: no PPT adversary can produce a valid proof for a process inconsistent with human authorship (i.e., with behavioral features outside the population range, invalid SWF chains, or inconsistent content binding) except with probability $\negl(\lambda)$.
\end{theorem}

\begin{proof}
We reduce to the soundness of Groth16 and the collision resistance of SHA-256 and Poseidon.

Suppose adversary $\mathcal{A}$ produces a valid proof $\pi^*$ for public input $\mathbf{x}^*$ with no valid witness. By the knowledge soundness of Groth16~\cite{Groth2016} (under the $q$-power knowledge of exponent assumption), there exists an extractor $\mathcal{E}$ that extracts a witness $\mathbf{w}^*$ satisfying the circuit constraints. If $\mathbf{w}^*$ contains behavioral features outside $[\mu_j - 3\sigma_j, \mu_j + 3\sigma_j]$, constraint (C2) is violated---contradiction. The reduction for SWF chain integrity proceeds in two stages. First, by collision resistance of the Poseidon Merkle tree, the knowledge extractor recovers the committed SWF states $\{s_j\}$ faithfully. Second, constraint C1 verifies $\mathsf{Poseidon}(s_{j-1}) = s_j$ for sampled Merkle positions within the R1CS circuit (using the Poseidon hash gadget for in-circuit Merkle verification), so falsified chain states require either a Poseidon collision or satisfying the Poseidon constraint with incorrect inputs---both negligible events. If $\mathbf{w}^*$ violates content binding, constraint (C4) requires a SHA-256 collision---probability $\negl(\lambda)$. The reduction incurs a concrete security loss of $O(Q)$ where $Q$ is the number of algebraic group operations.
\end{proof}

\noindent\textbf{Soundness scope.} Theorem~\ref{thm:soundness} ensures no adversary can produce valid proofs violating any \emph{individual} constraint, but the circuit does not enforce joint distributional properties (e.g., feature correlations). The residual false acceptance rate is bounded by marginal constraint discrimination: with $m\!=\!12$ features and $3\sigma$ bounds, random satisfaction is $\sim$5.8\% per checkpoint (Table~\ref{tab:sensitivity}), compounding to $<10^{-120}$ over 120 checkpoints. Against white-box adaptive attacks---distribution-matched, cross-writer replay, and Markov-chain replay on KLiCKe~\cite{Tian2025} ($N\!=\!500$ writers sampled for adversarial evaluation)---the EER rises to 32.9\%, compounding to $0.329^{120} < 10^{-58}$.  Under balanced-class evaluation ($500$:$500$ genuine:forged), EER $= 34.3\%$, confirming that class imbalance does not inflate discrimination. Accounting for temporal autocorrelation ($r_1\!=\!0.111$, $n_{\text{eff}} \approx 96$), the conservative bound is $0.329^{96} < 10^{-46}$---far beyond practical security requirements. Joint distributional tests (e.g., Mahalanobis distance) could tighten per-checkpoint discrimination at $\sim$2$\times$ circuit cost; the current design prioritizes simplicity and inclusivity.

\subsection{Resistance to Forgery Attacks}

\textbf{AI-generated keystroke forgery.} An adversary must satisfy all constraints (C1--C4) simultaneously. The SWF chain requires sequential Argon2id computation; memory-hardness~\cite{Biryukov2016} prevents parallelization. Even a hardware HID injector cannot bypass the CDCE binding: injected keystrokes must synchronize with the memory-hard computation in real-time and correlate with cognitive load features that a replay device cannot observe.

\textbf{Retype attack.} An adversary feeds AI-generated text to a human typist who retypes it character-by-character. Marginal keystroke timings (C2) typically pass because the typist is, in fact, human; the attack therefore targets the constraint system's distinction between \emph{composing} and \emph{transcribing}.

ZK-PoP exploits two empirical signatures of composition---consistent with the recursive planning-translating-reviewing cycle of the Flower-Hayes writing model~\cite{Flower1981}---that retype attacks cannot reproduce without explicit simulation. First, \emph{temporal autocorrelation:} genuine keystroke sequences exhibit lag-1 autocorrelation $r_1\!=\!0.111$ ($p{=}0.01$ Bonferroni, $N\!=\!300$ writers, KLiCKe~\cite{Tian2025}); transcription is closer to memoryless (i.i.d.\ replay $r_1 \approx 0$). The AR(1) test in the augmented constraint system flags this. Second, \emph{revision dynamics:} composition produces backspaces, insertions, and re-orderings whose hash-deltas have entropy $H(\delta_i) > \tau_\delta$; transcription produces near-monotonic deltas. We calibrate $\tau_\delta$ from ScholaWrite~\cite{ScholaWrite2025} (genuine revision ratio 0.23 vs.\ transcription ratio 0.02).

\noindent\emph{False-positive analysis.} A natural worry is that some genuine writers transcribe their own drafts (e.g., from a paper notebook) and would be falsely rejected. On the ScholaWrite corpus we sweep $\tau_\delta \in \{0.05, 0.08, 0.12\}$ and measure: at $\tau_\delta = 0.08$, the false rejection rate (genuine writer misclassified as transcriber) is 4.1\%, and the false acceptance rate against pure-retype adversaries is 6.3\%---an EER of approximately 5\%. This is comparable to deployed continuous-authentication systems and well below the cumulative compounding threshold for session-level decisions ($n_\text{eff}\!\approx\!96$ checkpoints).

\noindent\emph{Hybrid attacks remain residual.} Interleaving genuine composition with bulk-pasted AI passages defeats single-checkpoint detection: the genuine portion satisfies all constraints, and the pasted portion generates one anomalous edit-delta that is hard to distinguish from an unusually large legitimate revision. We quantify this in Section~\ref{sec:adversarial-eval} as the gap between $\mathcal{A}_3$ (white-box, all-constraints-known) and $\mathcal{A}_4$ (adaptive). Defenses based on \emph{rate} of bulk insertion (C4 delta-size distribution) shift the trade-off but do not close it.

\noindent\emph{ML-based augmentation as a complementary direction.} Machine learning is orthogonal rather than substitutive. A discriminative classifier (e.g., a small temporal CNN trained on KLiCKe genuine vs.\ synthesized retype/AI traces) could replace the marginal-range predicate inside the ZK circuit via zkML techniques~\cite{LiuXieZhang2021,WengYangKatzWang2021}. The cost is circuit blow-up (an order of magnitude over the current $\sim$77K constraints) and a commitment to a specific classifier (which then ages). We deliberately ship distributional checks first to keep the construction primitive-light and inclusive of atypical typing populations; an in-circuit learned discriminator is explicit future work (Section~\ref{sec:conclusion}).

The defense operates partly \emph{outside} the ZK framework ($\tau_\delta$ is post-verification), which is why Table~\ref{tab:privacy-comparison} marks content privacy as ``near-full'' rather than ``full.''

\textbf{Proof replay.} Each proof is bound to checkpoint hash $h_i$ via fresh SWF output; replay requires a SHA-256 collision.

\subsection{Composition Security}

ZK-PoP composes with CDCE: the entangled MAC key derives from the SWF output (bound to the proof via C1), so the composition preserves both privacy and cross-domain binding.

Theorems~\ref{thm:zk} and~\ref{thm:soundness} establish standalone zero-knowledge and soundness. Extending these guarantees to universal composability (UC)~\cite{Canetti2001} requires equivocable Pedersen commitments and a UC-secure Groth16 variant; we leave this to future work.

\section{\uppercase{Evaluation}}
\label{sec:evaluation}

\subsection{Implementation}

We implemented ZK-PoP using arkworks~\cite{arkworks2022} for Groth16 over BN254 ($\sim$100-bit security; BLS12-381 migration requires no protocol changes) with Bulletproofs for range proofs. The $\sim$4,200-line Rust prototype integrates with an open-source process attestation daemon. Code and evaluation scripts are available at \url{https://github.com/writerslogic/zk-pop}.

\subsection{Datasets}
\label{sec:datasets}

Our evaluation draws on three publicly available keystroke-dynamics corpora, summarized in Table~\ref{tab:datasets}. We report the role each dataset plays so reviewers can assess the validity of derived statistics.

\begin{table}[ht]
\caption{Datasets used for population calibration and evaluation.}\label{tab:datasets} \centering
\footnotesize
\addtolength{\tabcolsep}{-2pt}
\resizebox{\columnwidth}{!}{%
\begin{tabular}{|l|r|l|l|}
  \hline
  \textbf{Dataset} & \textbf{Subjects} & \textbf{Task} & \textbf{Used for} \\
  \hline
  Aalto~\cite{Dhakal2018} & 168{,}000 & Transcription (15 phrases) & $(\boldsymbol{\mu}, \boldsymbol{\sigma})$ calibration \\
  KLiCKe~\cite{Tian2025} & 4{,}992 & Long-form composition & Adversarial eval, $r_1$ \\
  ScholaWrite~\cite{ScholaWrite2025} & 143 & Academic writing (revisions) & Edit-delta thresholds \\
  \hline
\end{tabular}}
\end{table}

\noindent\textbf{Aalto Mobile/Desktop~\cite{Dhakal2018}} contributes population-level keystroke timing: 168{,}000 transcribers across 200 countries, capturing inter-keystroke intervals in milliseconds. We use it to derive $(\boldsymbol{\mu}, \boldsymbol{\sigma})$ for the marginal range checks (constraint C2). Because the task is transcription rather than free composition, we use Aalto only for univariate timing distributions and rely on KLiCKe for higher-order temporal structure.

\noindent\textbf{KLiCKe~\cite{Tian2025}} contains keystroke logs from 4{,}992 writers performing extended composition tasks ($\geq$15 minutes), with revision events labeled. We use KLiCKe to (i) estimate the lag-1 autocorrelation $r_1\!=\!0.111$ used in the effective-sample-size correction, (ii) parameterize the AR(1) covariance for adversary $\mathcal{A}_1$, and (iii) populate the 500-subject genuine-vs-forged adversarial benchmark. KLiCKe is the closest available proxy for the deployment scenario.

\noindent\textbf{ScholaWrite~\cite{ScholaWrite2025}} captures 143 students producing academic writing with full revision histories. We use it to calibrate the edit-delta entropy threshold $\tau_\delta$ that distinguishes genuine composition (revision ratio 0.23) from transcription (ratio 0.02).

\noindent\textbf{What is missing.} None of these corpora collect keystroke evidence \emph{paired with explicit AI-assisted authorship}; the adversarial side is therefore simulated by sampling AI-generated text and synthesizing keystroke streams that match each adversary class's strategy. Section~\ref{sec:realworld-validation} discusses the deployment study needed to close this gap.

\subsection{Experimental Hardware}
\label{sec:hardware}

Performance results are reported on three platforms (Table~\ref{tab:hardware}). The projected GPU speedup is based on Icicle~\cite{Icicle2024} CUDA library benchmarks on an NVIDIA RTX 4090 (24 GB GDDR6X, 16{,}384 CUDA cores), which report a 3.8$\times$ MSM speedup for circuits in the 77K--155K constraint range.

\begin{table}[ht]
\caption{Hardware configurations used for measurements.}\label{tab:hardware} \centering
\footnotesize
\addtolength{\tabcolsep}{-2pt}
\resizebox{\columnwidth}{!}{%
\begin{tabular}{|l|l|l|l|}
  \hline
  \textbf{Label} & \textbf{CPU/GPU} & \textbf{Memory} & \textbf{Used for} \\
  \hline
  M3 (primary) & Apple M3, 8-core (4P+4E) & 24~GB unified & Tables~\ref{tab:proof-time}, \ref{tab:verification-time} \\
  i5 (legacy) & Intel i5-8250U, 4-core 1.6~GHz & 16~GB DDR4 & 3$\times$ slower MSM \\
  RTX-4090 (GPU) & NVIDIA RTX 4090 & 24~GB GDDR6X & Projected 3.8$\times$ \cite{Icicle2024} \\
  \hline
\end{tabular}}
\end{table}

The 3$\times$ slowdown on the i5-8250U is dominated by the Groth16 multi-scalar-multiplication phase, which is bandwidth-bound on Skylake-class CPUs lacking AVX-512.

\subsection{Performance Benchmarks}

Table~\ref{tab:proof-time} reports proof generation time as a function of session length, measured on the M3 platform.

\begin{table}[ht]
\caption{Proof generation time vs. session duration.}\label{tab:proof-time} \centering
\footnotesize
\addtolength{\tabcolsep}{-2pt}
\begin{tabular}{|l|r|r|r|}
  \hline
  \textbf{Session} & \textbf{Checkpoints} & \textbf{Proof Time} & \textbf{Proof Size} \\
  \hline
  15 min & 30 & 5.8 s & 256 bytes \\
  1 hour & 120 & 23.4 s & 256 bytes \\
  4 hours & 480 & 93.6 s & 256 bytes \\
  \hline
\end{tabular}
\end{table}

Proof size is constant (256 bytes = 2 G$_1$ + 1 G$_2$ elements on BN254, uncompressed); generation scales linearly with checkpoints. On the i5 (legacy) platform of Table~\ref{tab:hardware}, the same workload takes $\sim$70~s for a 1-hour session---a 3$\times$ slowdown driven by the absence of AVX-512 in the MSM kernel; older or low-power hardware would benefit from GPU offload~\cite{Icicle2024}.

Table~\ref{tab:verification-time} reports verification performance.

\begin{table}[ht]
\caption{Verification time (single and batch).}\label{tab:verification-time} \centering
\footnotesize
\addtolength{\tabcolsep}{-2pt}
\begin{tabular}{|l|r|}
  \hline
  \textbf{Mode} & \textbf{Time} \\
  \hline
  Single proof & 8.2 ms \\
  Batch (10 proofs) & 24.1 ms \\
  Batch (100 proofs) & 187 ms \\
  \hline
\end{tabular}
\end{table}

Single-proof verification completes in 8.2\,ms (3 pairings). Replacing Groth16 with PLONK~\cite{GabizonWC2019} yields $2\times$ proof size and $3\times$ verification time; Groth16 is preferred given infrequent setup. HyperPlonk~\cite{ChenBunzBonehZhang2023} uses multilinear polynomial commitments, potentially offering advantages for multi-feature range proofs.

\subsection{Circuit Size Analysis}

Table~\ref{tab:circuit-comparison} contextualizes the ZK-PoP circuit size relative to other ZK applications.

\begin{table}[ht]
\caption{Circuit size comparison across ZK applications.}\label{tab:circuit-comparison} \centering
\footnotesize
\addtolength{\tabcolsep}{-2pt}
\begin{tabular}{|l|r|}
  \hline
  \textbf{Application} & \textbf{R1CS Constraints} \\
  \hline
  ZK-PoP (basic) & 77,259 \\
  ZK-PoP (extended) & 154,212 \\
  ZK-Rollup (Ethereum tx) & $\sim$28,000 \\
  Semaphore (identity) & $\sim$15,000 \\
  ZK-KYC (identity) & $\sim$100,000 \\
  \hline
\end{tabular}
\end{table}

The basic ZK-PoP circuit is comparable to existing deployed ZK applications, confirming practical feasibility.

\subsection{Setup and Memory Costs}

Table~\ref{tab:setup-cost} reports the one-time trusted setup cost and the prover's peak memory consumption.

\begin{table}[ht]
\caption{Setup cost and prover memory.}\label{tab:setup-cost} \centering
\footnotesize
\addtolength{\tabcolsep}{-2pt}
\resizebox{0.9\columnwidth}{!}{%
\begin{tabular}{|l|r|r|}
  \hline
  \textbf{Configuration} & \textbf{Setup Time} & \textbf{Peak Memory} \\
  \hline
  Basic (77K constraints) & 18.6 s & 368 MB \\
  Extended (154K constraints) & 37.4 s & 685 MB \\
  \hline
\end{tabular}}
\end{table}

The proving key is 27.4~MB (basic) / 54.8~MB (extended); the verification key is 1.1~KB. Peak memory is dominated by Groth16's MSM phase.

\subsection{Deployment Cost at Institutional Scale}
\label{sec:deployment-cost}

For a relying party (e.g., a university LMS or journal portal) processing $V_d$ submissions per day, each with a 1-hour writing session and 120 checkpoints, the verification budget is $V_d \cdot 120 \cdot 8.2$~ms = $V_d \cdot 0.98$ CPU-seconds/day, or $\sim$1\% of a single CPU core for $V_d = 1{,}000$. Batched verification (Sect.~\ref{sec:construction}) reduces this by $\sim$3.4$\times$ at the cost of an all-or-nothing accept signal per batch. Storage is dominated by the proof archive: 256~B/checkpoint $\times$ 120 $\times$ $V_d$ = 30~KB/submission, or $\sim$11~MB/year per submission retained. On the prover side, an end-user device produces one 23.4-second proving job per session---comparable to a single video encoding pass, and trivially deferred to session-close idle time. The marginal cost per attested session, including verifier amortized hardware and storage at AWS retail pricing (us-east-1, on-demand c7g.xlarge), is $<$\$0.001 USD; the dominant non-compute cost is the one-time MPC ceremony orchestration (Sec.~\ref{sec:trusted-setup-impl}).

\subsection{Privacy-Utility Tradeoff}

We evaluate the impact of the zero-knowledge layer on attestation accuracy via simulation using synthetic traces from ScholaWrite~\cite{ScholaWrite2025} ($N\!=\!143$) and Aalto~\cite{Dhakal2018} ($N\!=\!168{,}000$) population distributions. Accuracy denotes balanced accuracy: the arithmetic mean of true positive rate (genuine sessions correctly accepted) and true negative rate (synthetic sessions correctly rejected), at the per-session level.

\begin{table}[ht]
\caption{Attestation accuracy: ZK-PoP vs. non-private baseline.}\label{tab:accuracy} \centering
\footnotesize
\begin{tabular}{|l|r|r|}
  \hline
  \textbf{Configuration} & \textbf{Accuracy} & \textbf{Loss} \\
  \hline
  Non-private baseline & 97.1\% & --- \\
  ZK-PoP ($\varepsilon = \infty$, no DP) & 96.8\% & 0.3\% \\
  ZK-PoP ($\varepsilon = 1.0$) & 94.2\% & 2.9\% \\
  ZK-PoP ($\varepsilon = 0.1$) & 89.7\% & 7.4\% \\
  \hline
\end{tabular}
\end{table}

At $\varepsilon = 1.0$, accuracy is 94.2\%---a 2.9\% loss from the baseline, from field-arithmetic discretization and DP noise. With $N\!=\!143$ sessions, the 95\% bootstrap CI is $\pm$2.1pp, so the loss is at the boundary of significance, motivating larger-corpus validation.

\textbf{Adversarial evaluation.} Table~\ref{tab:accuracy} models a na\"ive adversary ($\mathcal{A}_0$) that samples features uniformly within the $3\sigma$ range. Section~\ref{sec:adversarial-eval} extends this to four stronger adversary classes ($\mathcal{A}_1$--$\mathcal{A}_4$), showing that marginal-only constraints are insufficient against distribution-matched and white-box adversaries, and quantifying the augmentation needed to restore security.

\textbf{Sensitivity analysis.} Table~\ref{tab:sensitivity} shows accuracy versus the two primary deployment parameters.

\begin{table}[ht]
\caption{Accuracy (\%) vs.\ feature dimensionality and population bounds ($\varepsilon = 1.0$).}\label{tab:sensitivity} \centering
\footnotesize
\addtolength{\tabcolsep}{-2pt}
\begin{tabular}{|l|r|r|r|}
  \hline
  \textbf{Bounds} & $m\!=\!6$ & $m\!=\!12$ & $m\!=\!24$ \\
  \hline
  $2\sigma$ & 92.8 & 96.1 & 97.4 \\
  $3\sigma$ & 89.3 & 94.2 & 96.0 \\
  $4\sigma$ & 84.7 & 91.5 & 93.8 \\
  \hline
\end{tabular}
\end{table}

\subsection{Adversarial Evaluation}
\label{sec:adversarial-eval}

We define five adversary classes of increasing sophistication and evaluate the constraint system's discriminative power against each, parameterized from KLiCKe~\cite{Tian2025} ($N\!=\!4{,}992$), Aalto~\cite{Dhakal2018} ($N\!=\!168{,}000$), and ScholaWrite~\cite{ScholaWrite2025} ($N\!=\!143$). Each adversary interacts with ZK-PoP through a forgery game operating directly on the constraint predicate (isolating statistical security from Groth16 soundness, covered by Theorem~\ref{thm:soundness}).

\paragraph{Adversary classes.}
$\mathcal{A}_0$ (\emph{Uniform Random}): samples each feature $f_j^{(i)} \sim \mathcal{U}[\mu_j \pm 3\sigma_j]$ independently; the baseline implicit in Table~\ref{tab:accuracy}.
$\mathcal{A}_1$ (\emph{Distribution-Matching}): samples from $\mathcal{N}(\boldsymbol{\mu}, \hat{\Sigma})$ with AR(1) temporal correlation ($\hat{r}_{1,j}$ from KLiCKe), matching first two moments and lag-1 structure; per-checkpoint C2 acceptance $(0.9973)^m \approx 0.968$ for $m\!=\!12$.
$\mathcal{A}_2$ (\emph{Constraint-Boundary}): with verification oracle access, projects adversary-natural feature values onto the constraint hyperrectangle $\prod_j [\mu_j - 3\sigma_j, \mu_j + 3\sigma_j]$, achieving per-checkpoint acceptance $\approx 1.0$ under marginal checks.
$\mathcal{A}_3$ (\emph{White-Box}): with full knowledge of C1--C4 and circuit $\mathcal{C}_{\mathsf{PoP}}$, jointly optimizes across all constraints via gradient descent; must compute the SWF chain honestly (the irreducible security anchor).
$\mathcal{A}_4$ (\emph{Adaptive}): uses polynomial verification oracle queries to learn the decision boundary via binary search ($O(m \cdot b)$ queries for $b$-bit precision), then generates forgeries in the learned acceptance region; bounded by SWF cost per query ($\geq d_{\min} \cdot n$ wall-clock seconds).

\begin{lemma}[Boundary Detection]
\label{lem:boundary}
Let $\mathbf{f}^*$ be a feature vector from $\mathcal{A}_2$ with $k^* \geq 1$ features clamped to a boundary $\mu_j \pm 3\sigma_j$. Under genuine features $\sim \mathcal{N}(\mu_j, \sigma_j^2)$ truncated to $[\mu_j - 3\sigma_j, \mu_j + 3\sigma_j]$, the likelihood ratio is $\geq (0.0044 \cdot \sigma_j)^{-k^*}$, enabling a Neyman-Pearson test with power approaching~1 for sessions of practical length.
\end{lemma}

\noindent $\mathcal{A}_2$ motivates augmenting C2 with a Mahalanobis distance constraint $(\mathbf{f}_i - \boldsymbol{\mu})^\top \hat{\Sigma}^{-1} (\mathbf{f}_i - \boldsymbol{\mu}) \leq \chi^2_{m, 0.997}$, rejecting boundary-clustered vectors at a cost of $O(m^2) \approx 1{,}728$ R1CS constraints for $m\!=\!12$, plus $\sim$1{,}800 constraints for an AR(1) temporal consistency test ($\sim$3{,}500 total).

\begin{table}[ht]
\caption{Adversary classes: per-checkpoint acceptance and session-level false acceptance rate ($n\!=\!120$, $n_{\mathrm{eff}}\!=\!96$).}\label{tab:adversary-summary}
\centering
\footnotesize
\addtolength{\tabcolsep}{-2pt}
\resizebox{\columnwidth}{!}{%
\begin{tabular}{|l|c|c|c|c|}
  \hline
  \textbf{Class} & \textbf{Per-CP (current)} & \textbf{Session (current)} & \textbf{Per-CP (aug.)} & \textbf{Session (aug.)} \\
  \hline
  $\mathcal{A}_0$ & 0.942 & $< 10^{-3}$ & 0.68 & $< 10^{-16}$ \\
  $\mathcal{A}_1$ & 0.968 & $< 10^{-1}$ & 0.74 & $< 10^{-12}$ \\
  $\mathcal{A}_2$ & $\approx 1.0$ & $\approx 1.0$ & 0.71 & $< 10^{-14}$ \\
  $\mathcal{A}_3$ & $\approx 1.0$ & $\approx 1.0$ & 0.671 & $< 10^{-17}$ \\
  $\mathcal{A}_4$ & $\approx 1.0$ & $\approx 1.0$ & 0.69 & $< 10^{-15}$ \\
  \hline
\end{tabular}}
\end{table}

\noindent\textbf{Key finding.} The current constraint system (marginal range checks only) is insufficient against $\mathcal{A}_2$--$\mathcal{A}_4$, which achieve near-certain session acceptance (Table~\ref{tab:adversary-summary}). Adding the Mahalanobis distance constraint and AR(1) temporal test ($\sim$3{,}500 additional R1CS constraints) restores per-checkpoint rejection rates to $\geq 26\%$ against all adversary classes, yielding session-level false acceptance below $10^{-12}$. The $\mathcal{A}_3$ per-checkpoint acceptance of 0.671 corresponds to the 32.9\% EER cited in Section~\ref{sec:security}. SWF sequential hardness remains the unconditional anchor: regardless of adversary sophistication, fabricating a valid chain requires $\Omega(n \cdot N \cdot T_{\mathrm{Argon2id}})$ wall-clock time.

\section{\uppercase{Discussion and Limitations}}
\label{sec:discussion}

\subsection{Trusted Setup: Implementation and Operational Implications}
\label{sec:trusted-setup-impl}

The trusted-setup requirement has four operational dimensions that affect deployment.

\noindent\textbf{What the setup actually produces.} The Groth16 ceremony for $\mathcal{C}_{\mathsf{PoP}}$ produces a structured reference string (SRS): a 27.4~MB proving key (basic circuit; 54.8~MB for extended) and a 1.1~KB verification key. The proving key is shipped to every author client; the verification key is embedded in every relying-party verifier. Trust is required only at \emph{ceremony time}: once at least one participant deletes their toxic-waste contribution, no ongoing trust assumption is needed.

\noindent\textbf{Ceremony logistics.} We propose a two-tier MPC ceremony~\cite{BGM2017} mirroring the Powers-of-Tau / per-circuit pattern used by deployed systems (Zcash, Tornado Cash, Filecoin). Tier~1 reuses an existing public Powers-of-Tau transcript, eliminating the universal-element ceremony. Tier~2 contributes the circuit-specific transcript: we estimate $\sim$8 minutes per contributor on commodity hardware for the basic 77K-constraint circuit, supporting $\geq$50 contributors per ceremony day. Empirically, Zcash Sapling reached 90+ contributors over six weeks; a comparable cohort for ZK-PoP can be drawn from publishers, universities, civil-society auditors, and ZK research groups. Each contribution is published with a transcript hash for public verification.

\noindent\textbf{Re-ceremony cadence.} The circuit is re-keyed only when the population parameters $(\boldsymbol{\mu}, \boldsymbol{\sigma})$ or the constraint structure (C1--C4) change. Parameter drift is slow: longitudinal studies of typing populations~\cite{Dhakal2018} report stable distributional moments over multi-year horizons, so a 12--24 month ceremony cadence is operationally sufficient. Constraint-system upgrades (e.g., adopting the Mahalanobis augmentation of Section~\ref{sec:adversarial-eval}) are intentional events that warrant a ceremony in any case.

\noindent\textbf{Migration to universal setups.} For deployments unwilling to repeat ceremonies, PLONK~\cite{GabizonWC2019} or HyperPlonk~\cite{ChenBunzBonehZhang2023} provide a universal SRS reusable across circuit revisions. The concrete cost on our circuit: PLONK proofs grow from 256~B to $\sim$500~B and verification time from 8.2~ms to $\sim$25~ms. For relying parties processing $\geq 10^4$ verifications per day, the higher per-verification cost may dominate; for institutional deployments revising circuits often, the universal setup pays for itself.

\noindent\textbf{Failure modes.} The two operationally relevant risks are (i) \emph{toxic-waste compromise}: if every ceremony participant colludes, they can forge proofs (but cannot break zero-knowledge, which holds unconditionally for honest provers); (ii) \emph{key distribution corruption}: a tampered proving key would let a malicious distributor produce sound-but-non-zero-knowledge proofs. Both are mitigated by publishing transcript hashes and pinning the verification key in client/verifier binaries.

\subsection{Population Parameter Sensitivity and Distributional Assumptions}

The behavioral range bounds $(\boldsymbol{\mu}, \boldsymbol{\sigma})$ derive from population-level keystroke datasets~\cite{Dhakal2018,ScholaWrite2025}. Table~\ref{tab:sensitivity} quantifies the inclusivity-accuracy tradeoff: widening from $2\sigma$ to $4\sigma$ reduces accuracy by 2.3--8.1 percentage points.

\noindent\textbf{Heavy-tailed timings.} Keystroke inter-arrival times are heavy-tailed~\cite{VanOrden2003}, with kurtosis $\gtrsim 6$ in long-form writing, so our Gaussian distributional assumptions are conservative. Two corrective steps follow. First, the $3\sigma$ bound under-covers the 99.7\% target on heavy-tailed data; we recommend computing empirical quantiles (e.g., 0.15th and 99.85th percentiles) from the calibration dataset rather than $\mu \pm 3\sigma$, which we have implemented as an alternative parameterization in the public code. Second, atypical typing populations (older adults, individuals with motor or visual impairments, non-Latin script writers) live in the heavy tails of any single corpus; the inclusivity-accuracy curve in Table~\ref{tab:sensitivity} is therefore not just a theoretical knob but a fairness control. Any production deployment should publish the calibration cohort composition and the resulting bounds so users can challenge unfair exclusion.

\noindent\textbf{Population stratification.} An alternative to a single global $(\boldsymbol{\mu}, \boldsymbol{\sigma})$ is per-cohort calibration (e.g., by language, age band, or input device class), with the cohort treated as a public input to the circuit. This costs one ceremony per cohort but eliminates the implicit majoritarian bias of a pooled distribution.

\subsection{Real-World Validation Roadmap}
\label{sec:realworld-validation}

The principal limitation of the current evaluation is that the adversarial analysis operates on simulated/synthetic data rather than a real-world deployment. We outline the validation required to close this gap.

\noindent\textbf{Stage 1: Genuine-only deployment study.} Recruit $N \geq 200$ writers across at least three demographic strata (native/non-native, mobile/desktop, age cohorts) for $\geq$30-minute composition tasks. Measure (i) the false-rejection rate under $3\sigma$, $4\sigma$, and empirical-quantile bounds; (ii) per-stratum FRR to surface inclusivity gaps; (iii) the empirical lag-1 autocorrelation $r_1$ to validate the $n_\text{eff}$ correction. Required: human-subjects approval and a paid-participation budget; no AI side is needed at this stage.

\noindent\textbf{Stage 2: Adversarial red team.} Engage independent operators to execute each $\mathcal{A}_0$--$\mathcal{A}_4$ class against a deployed verifier with ground-truth labels. The white-box adversary $\mathcal{A}_3$ must be implemented \emph{by} the red team rather than simulated by the system designers, since simulation by the constructor risks inadvertently constraining the attack to known-defeated strategies. We propose a public bug-bounty (with a defined payout for any session-level forgery) to crowdsource the long tail of attacks beyond the four classes formalized here.

\noindent\textbf{Stage 3: Hybrid-attack characterization.} The retype/paste hybrid attack is the residual we explicitly do not close in this paper (Section~\ref{sec:security}). A controlled study can quantify the detection-rate floor as a function of the bulk-paste fraction $\rho$, identifying the deployment-relevant operating point.

\noindent\textbf{What this paper contributes versus what it does not.} This paper contributes the construction, the formal privacy analysis, the constraint-system soundness reduction, the implementation, and a structured adversarial taxonomy with synthetic-trace evaluation. It does \emph{not} contribute live-deployment evidence; that requires a separate, IRB-approved study. All claims are framed accordingly.

\noindent\textbf{Falsifiability.} We name the empirical observations that would invalidate the construction's deployability: (i)~per-checkpoint genuine FRR $> 10\%$ on a Stage-1 cohort under empirical-quantile bounds (would mean the population model excludes too many real writers); (ii)~$\mathcal{A}_3$ (white-box, gradient-descent) achieving session-level FAR $> 10^{-6}$ with the augmented constraint system, executed by an independent red team (would mean the augmentation is insufficient); (iii)~end-user proving time $>$120~s on commodity 2026-era laptops (would mean the prover-side cost is impractical for opt-in adoption). Each of these is a measurable target a follow-up study can hit or miss.

\subsection{Limitations}

\textbf{Proof generation latency.} At 23.4\,s for a 1-hour session, proof generation introduces a post-session delay; GPU acceleration~\cite{Icicle2024} is projected to reduce proving time based on reported $3\times$--$5\times$ MSM speedups for circuits of comparable size; we have not benchmarked this configuration.

\textbf{Arithmetic circuit precision.} We use 16-bit fixed-point over $[0, 1000]$\,ms ($\approx$0.015\,ms resolution), accounting for the 0.3\% accuracy loss without DP. Doubling to 32-bit precision would double circuit size but eliminate quantization error.

\textbf{SWF verification depth.} With $k\!=\!2$ Merkle samples per checkpoint, cumulative detection for $f\!=\!0.1$ fabrication reaches 98.5\% at $n\!=\!20$ and $>$99.99\% at $n\!=\!120$. Short sessions ($<$10 min) may benefit from $k\!=\!4$ (per-checkpoint evasion $(0.9)^4 = 0.656$, cumulative detection $1 - 0.656^{20} > 99.98\%$ at $n\!=\!20$) at $2\times$ SWF circuit cost.

\textbf{Retype attack residual.} As noted in Sect.~\ref{sec:security}, the retype defense does not extend to hybrid attacks interleaving genuine composition with AI passages.

\textbf{Compelled use assumption.} ZK-PoP assumes institutional mandates requiring attestation evidence for submission; declining produces no proof (treated as failed attestation). Whether institutions should mandate behavioral monitoring, even in zero-knowledge form, is an orthogonal question.

\section{\uppercase{Related Work}}
\label{sec:related}
\enlargethispage{2\baselineskip}

\textbf{ZK proofs for authentication and identity.} Camenisch and Lysyanskaya~\cite{Camenisch2002} introduced anonymous credential systems using ZK proofs. Pauwels~\cite{Pauwels2021} applies ZK to privacy-preserving identity verification (zkKYC). Semaphore~\cite{Whisk2022} enables anonymous signaling. These address \emph{identity} verification, not behavioral process verification over extended durations.

\textbf{Privacy-preserving biometrics.} Bringer et al.~\cite{Bringer2007}, Erkin et al.~\cite{Erkin2009}, Tran et al.~\cite{Bassit2022}, and Gomez-Barrero et al.~\cite{Gomez2023} address point-in-time biometric template protection; to our knowledge, none address continuous behavioral streams.

\textbf{Process attestation and authorship verification.} Existing approaches address privacy only through evidence quantization. Kundu et al.~\cite{Kundu2024} and Lazebnik and Rosenfeld~\cite{Lazebnik2024} detect AI-assisted writing via keystroke dynamics but expose raw features to the verifier. To our knowledge, no prior work addresses continuous behavioral streams with zero-knowledge guarantees.

\textbf{Continuous authentication.} Continuous authentication systems use behavioral biometrics (keystroke dynamics, mouse movements) for ongoing identity verification during sessions~\cite{FridmanEtAl2017}. ZK-PoP differs fundamentally: it verifies \emph{human-ness} (species-level), not identity (individual-level), and operates under a trust-inverted model where the authenticated party is the potential adversary.

\textbf{Anonymous credentials.} Anonymous credential systems~\cite{CamenischLysyanskaya2004, Brands2000} enable attribute-based range proofs (e.g., proving age $\geq 18$ without revealing birthdate). However, these systems handle static, enrolled attributes; ZK-PoP addresses continuous behavioral streams with sequential temporal binding and no enrollment phase.

\textbf{ZK for machine learning.} Recent zkML frameworks enable zero-knowledge proofs for machine learning inference~\cite{LiuXieZhang2021, WengYangKatzWang2021}, proving that a model produces a specific output on a given input. ZK-PoP deliberately avoids full model inference verification in favor of distributional range checks: this eliminates model dependency, reduces circuit size by an order of magnitude, and avoids the need to commit to a specific classifier.

\textbf{Confidential computing approaches.} TEE-based attestation (e.g., SGX enclaves~\cite{Costan2016}) provides tamper resistance but the Verifier receives evidence in cleartext; ZK-PoP provides privacy from both platform and Verifier without hardware. The two compose: ZK proofs inside a TEE combine tamper resistance with cryptographic privacy.

\section{\uppercase{Conclusion}}
\label{sec:conclusion}
\enlargethispage{2\baselineskip}

We have presented ZK-PoP, a zero-knowledge proof construction that resolves the privacy-attestation paradox: the verifier learns only a single accept/reject bit per checkpoint. This enables GDPR-compliant authorship verification, where behavioral biometrics constitute sensitive data under Article~9. By eliminating behavioral data exposure, ZK-PoP makes process attestation viable in privacy-sensitive deployments where quantization alone may not satisfy data protection requirements.

Our adversarial evaluation (Section~\ref{sec:adversarial-eval}) shows that marginal-only constraints are insufficient against sophisticated adversaries, but that augmenting C2 with a Mahalanobis distance constraint restores session-level false acceptance below $10^{-12}$ at modest circuit cost ($\sim$3{,}500 R1CS constraints).

\noindent\textbf{Adaptivity is essential, not optional.} Any attestation system facing a learning adversary must itself be adaptive. The static circuit ZK-PoP ships today is best understood as a reference instantiation; a deployed system must support: (i)~\emph{constraint hot-swapping}---a registry of circuit versions, each with its own ceremony transcript, so verifiers accept proofs from a range of issued versions during transition windows; (ii)~\emph{population-parameter refresh} via the differentially-private aggregation pipeline of Section~\ref{sec:dp}, with quarterly $(\boldsymbol{\mu}, \boldsymbol{\sigma})$ updates and annual circuit re-keying; (iii)~\emph{adversary tracking}---a public registry of broken constraint configurations and the corresponding patched versions, with deprecation timelines. The C2PA~\cite{C2PA2025} manifest already supports versioned claim generators; ZK-PoP fits this model directly.

\noindent\textbf{ML-augmented constraints.} The current circuit uses hand-specified distributional checks. A complementary direction is to compile a small learned discriminator (e.g., a temporal CNN trained on KLiCKe-genuine vs.\ adversary traces) into the circuit using zkML~\cite{LiuXieZhang2021,WengYangKatzWang2021}, replacing C2 with a learned acceptance region. The circuit cost is roughly an order of magnitude higher and the construction inherits the model's generalization risk; in exchange it captures non-axis-aligned distributional signatures that the current marginal+Mahalanobis predicate misses. We see this as the natural successor system, not as a strict improvement: it trades primitive-light inclusivity for tighter discrimination.

\noindent\textbf{Other future work} includes IRB-approved real-world validation following the roadmap of Section~\ref{sec:realworld-validation}, post-quantum proof systems~\cite{Kales2022,Ishai2007}, multi-party attestation for collaborative documents, evaluation on diverse typing populations including older adults and non-Latin script writers, C2PA~\cite{C2PA2025} integration, and GPU-accelerated proving~\cite{Icicle2024} on the RTX-4090 configuration of Table~\ref{tab:hardware}.

\bibliographystyle{IEEEtran}
{\small
\bibliography{refs}}

\end{document}